\documentclass[10pt,twocolumn,twoside]{IEEEtran}\IEEEoverridecommandlockouts 			

\usepackage{hyperref}
\usepackage{graphics} 
\usepackage{epsfig} 
\usepackage{times} 
\usepackage[tbtags]{amsmath} 
\usepackage{amssymb}  
\usepackage{amsthm}
\usepackage{mathtools}
\usepackage{cite}
\usepackage{hyperref}
\usepackage{algorithm}
\usepackage{algpseudocode}
\usepackage{xcolor}
\usepackage{tikz}
\usepackage{graphicx}
\newtheorem{theorem}{\textbf{Theorem}}
\newtheorem{proposition}{\textbf{Proposition}}
\newtheorem{lemma}{\textbf{Lemma}}

\newtheorem{definition}{Definition}

\newtheorem{remark}{\textbf{Remark}}
\newtheorem{assumption}{\textbf{Assumption}}

\usepackage{subcaption}
\usetikzlibrary{shapes}
\usepackage{float}
\usetikzlibrary{calc}
\usepackage{enumerate}

\newcommand{\trf}{\mathcal{G}}
\newcommand{\nbr}{\mathcal{N}}

\newcommand{\real}{\mathbb{R}}
\newcommand{\mean}{\mathbb{E}}
\title{Topology Learning in Radial Dynamical Systems with Unreliable Data}
\author{Venkat Ram Subramanian, Deepjyoti Deka, Saurav Talukdar, Andy Lamperski, Murti Salapaka
\thanks{
V.R.S., S.T., A.L., M.S. are with the Department of Electrical and Computer
   Engineering, University of Minnesota, Minneapolis, MN 55455, USA. {\tt\small subra148@umn.edu, sauravtalukdar@umn.edu, alampers@umn.edu,  murtis@umn.edu} D.D. is with the Theoretical Division at Los Alamos National Laboratory, New Mexico, USA. {\tt\small deepjyoti@lanl.gov}}
   \thanks{Work supported in part by NSF CMMI 1727096. }
}
\begin{document}
\maketitle
\begin{abstract}
Many complex engineering systems admit 
bidirectional and linear couplings between their agents. Blind and passive methods to identify such influence
pathways/couplings from data are central to many applications. However,
dynamically related data-streams originating at different sources are
prone to corruption caused by asynchronous time-stamps of different
streams, packet drops and noise. Such imperfect information may be present in the entire observation period, and hence not detected by change-detection algorithms that require an initial clean observation period. Prior work has shown that spurious links are inferred in the graph structure due to the corrupted data-streams, which prevents consistent learning. In this article, we provide a novel approach to
detect the location of corrupt agents as well as present an algorithm to learn the structure of radial dynamical systems despite corrupted data streams. In particular, we show that our approach provably learns the true radial structure if the unknown corrupted nodes are at least three \emph{hops} away from each other. Our theoretical results are further validated in test dynamical network. 
\end{abstract}
\section{Introduction}
Networks provide an effective representation framework to analyze the interactions in complex systems. Network representations are widely used in a various fields like neuroscience \cite{bassett2017network}, social networks \cite{scott1988social}, power grid \cite{deka2018structure} to name a few. Learning the network representation can provide insights into the analysis of system behavior, help identify critical links, detect faults, and optimize flows. Initial research on learning the network representation involved considering the outputs of the system as  random variables \cite{jordan1998learning} . However, such an approach requires independence among time-lagged observations and is not applicable in the presence of dynamics in the system and with the availability of high (time) resolution observations from the system. 

Two distinct paradigms have emerged, active network inference \cite{he2008active} and passive network inference \cite{buntine1996guide}. In the former, the value of an output is set to a fixed value and the impact on other outputs are analyzed to infer the network structure. Whereas, only the time series observations from the system are utilized to infer the network structure without any active intervention in the latter. In this article, we will focus on a passive approach to inference of the network structure. In fact, in many applications like the stock market, and power grid,  active intervention is much costlier and often impermissible compared to passive data collection.

In this article, we consider dynamical systems where every coupling between agents/nodes is considered to be a bi-directional. There are several physical systems, especially flow driven systems like power grid networks \cite{talukdar2017learning}, heat transfer networks \cite{pass2018thermodynamic}, fluid flow networks, and others like networks of oscillators
\cite{nunez2015synchronization} and consensus networks
\cite{patel2017distributed},  where a notion of a directed edge is insufficient to capture the system dynamics, and requires bi-directed edges to correctly represent the influence between two nodes.

Inferring the network representation from time series measurements for complex dynamical system has recently gained interest in controls community \cite{goncalves2008necessary}, \cite{yuan2011robust}, \cite{gevers2017on}, \cite{weerts2018identifiability}, \cite{van2019necessary}.  
In the context of Linear Time Invariant (LTI) systems,  \cite{MatSal12} use multivariate Wiener filters to infer the network structure in a passive manner in the context of LTI systems. Here, the authors infer the moral graph of the system, which has spurious edges between graph nodes that are two-hops away. 
In \cite{talukdar2017Exact}, the authors show that for bi-directed LTI system networks with a radial network topology, the spurious edges in the moral graph can be eliminated using graphical separation rules. Similarly, phase -based results of the estimated Wiener filters have been shown to enable removal of spurious edges in non-radial bi-directional LTI systems \cite{talukdar2020physics}. The problem of learning polytree structures has been studied in \cite{etesami2016learning} and \cite{sepehr2019blind}. 
The authors provide guarantees of a consistent reconstruction. All the above work assume that the measurements are uniformly sampled and are available without any non ideal aspect like packet drops or random delays. Often, the data-streams in large systems are not immune to effects of noise \cite{stankovic2018distributed}, asynchronous sensor clocks \cite{cho2014survey}, \cite{cavraro2019dynamic} and packet drops \cite{leong17sensor},\cite{zhou2018distributed}. In \cite{SLS17network} focusing on directed networks with linear time-invariant (LTI) interactions, authors provided characterization of the extent of spurious
links that can appear due to data-corruption in the moral graph. However, little is known if these spurious edges can be eliminated to infer the exact network structure even in the presence of corruptions in the data streams, thus establishing consistency guarantees. In \cite{SLS2019corruption}, focusing on bi-directional networks, it is shown that the location of corrupt nodes can be detected by combining tools from information theory and graph theory. However, a method to eliminate spurious edges was not presented. 

 \textit{Our contribution:} In this article, the objective is to determine the exact network representation of \emph{radial} bi-directional LTI systems, using passive means from corrupt data-streams. We show that for radial bi-directed network of LTI systems where corrupt nodes are located deep in the network, at least three hops away from the leaf nodes, the spurious edges owning to data corruption can be eliminated and the the exact network structure can be inferred. We present novel topological characterizations and phase-based properties to determine the exact location of corruptions. Finally, we propose an algorithm called `hide and learn' to determine the exact topology generating the time series observations. To this, we follow a similar topology learning algorithm, presented in our prior conference paper \cite{talukdar2018topology} that considered hidden nodes. However in \cite{talukdar2018topology}, there was a tighter assumption on the distance between hidden nodes restricting them to be at least four hops away from each other and the measurements were assumed to be perfect. Moreover, rigorous proofs were not presented. Here, we consider time-series with imperfect information and relax the assumption on the location of corrupt nodes, and provide rigorous proofs to our results.
 
The preliminary section ~\ref{sec:prelim} describes the generative model, the graphical representation for the network and reviews earlier work on LTI network identification using power spectra. Section ~\ref{sec:uncertain} describes the data corruption models and its effect on structure inference. In Section ~\ref{sec:corruptionDetection}, we present main result to determine the location of corrupt nodes. The exact topology learning algorithm is presented in Section ~\ref{sec:exactLearn}. Simulation results are provided in Section ~\ref{sec:sim} and finally, a conclusion is provided in Section ~\ref{sec:conclude}. 
\subsection*{Notation}
\noindent $Y$ denotes a vector with $y_i$ being $i^{th}$ element of $Y.$ \\
$i - j$ denotes an undirected edge between nodes $i$, $j$ in an undirected graph while $i\to j$ denotes a directed edge from $i$ to $j$ in a directed graph.\\
If $M(z)$ is a transfer function matrix, then $M(z)^* = M(z^{-1})^T$ is the conjugate transpose.$M(i,j)$ denotes the matrix entry at $i^{th}$ row and $j^{th}$ column.\\
$\mathbb{E}[\cdot ]$ denotes expectation operator.\\
$R_{XY}(k):=\mathbb{E}[X[n+k]Y[n]]$ is the cross-correlation function of jointly wide-sense stationary(WSS) processes $X$ and $Y$. If $Y=X$ then $R_{XX}(k)$ is called the auto-correlation.\\
$\Phi _{XY}(z):=\mathcal{Z}(R_{XY}(k))$ represents the cross-power spectral density while $\Phi _{XX}(z):=\mathcal{Z}(R_{XX}(k))$ denotes the power spectral density(PSD).
$\mathcal{Z}(\cdot )$ is the Z-transform operator.\\
$b_i$ represents the $i^{th}$ element of the canonical basis of
$\real ^n$.
\section{Preliminaries}\label{sec:prelim} In this section the generative model and the generative graph that represents the networked system are presented. 
\subsection{Generative Model}\label{sec:GenModel}
Consider $N$ agents that interact over a
network. Consider the following continuous time linear dynamics for each agent $i \in \{1,\cdots N\}$:
\begin{equation}\label{eq:ctDyn}
    \sum _{m=1}^n a_{m,i}\frac{d^mx_i}{dt^m}  = \sum_{j=1,j\ne i}^Nb_{ij}(x_j(t)-x_i(t)) + w_i(t),
\end{equation} where the process $w_i(t)$ is considered to be zero mean WSS process innate to agent $i$ and thus $w_i$ is independent of $w_j$ if $i\not=
j.$  Thus, the power spectral density (PSD) of $w=(w_1,w_2,\dots , w_N)^T,$ $\Phi_w(z)$ is a diagonal matrix. Above, $a_{m,i}$, $b_{ij}\in \real $. We assume the signals are bounded in a mean-square
sense: $\mean [\parallel  x_i[t]\parallel  ^2] <\infty $ and $\mean
[\parallel w_i[t] \parallel ^2] < \infty$.
After discretization and taking $z$ transform we obtain the following:
\begin{equation}\label{eq:Dyn}
S_i(z)x_i(z) = \sum_{j=1,j\ne i}^Nb_{ij}x_j(z) + w_i(z) \ \ \mbox{for
}i=1,\ldots, N. 
\end{equation}
Here $s_i(z) $ is a transfer function obtained due to the derivatives of $x_i(t)$.  Rewriting the above equation we obtain:
\begin{equation}\label{eq:dim}
x_i(z) = \sum_{j=1,j\ne i}^N\trf_{ij}(z)x_j(z) + e_i(z) \ \ \mbox{for
}i=1,\ldots, N. 
\end{equation} 
Here, $\trf_{ij}(z)=\frac{b_{ij}}{S_i(z)}$, $e_i(z)=\frac{w_i(z)}{S_i(z)}$.

Compactly,  \eqref{eq:dim} is equivalent to 
\begin{equation}
\label{eq:dimCompact}
x = \trf (z) x + e,
\end{equation}
where $x=(x_1(z),x_2(z),\dots,x_N(z))^T$ and $e=(e_1(z),e_2(z),\dots,e_N(z))^T$ and $\trf(i,j)=\trf_{ij}(z)$.
We call the pair $(\trf(z),e)$  \textit{generative model}. We consider generative models such that $\trf_{ij}(z)\ne 0$ and $\trf_{ji}(z)\ne 0$. Such models are prevalent in linearized models of engineering systems operating around an equilibrium point. For example, consider swing dynamics for power systems and heat transfer dynamical systems.
%
\begin{figure}
\centering
\subcaptionbox{
  \label{fig:gengraph} 
    Generative Graph $G$} 
 {\begin{subfigure}{0.49\columnwidth}
\centering
 \begin{tikzpicture}[scale=0.65]
 \tikzstyle{vertex}=[circle,fill=none,minimum size=10pt,inner sep=0pt,thick,draw]
 \tikzstyle{pvertex}=[star,star points=14,fill=white,minimum size=10pt,inner sep=0pt,thick,draw]
  \node[vertex] (n1) {$1$};

  \node[vertex] (n2) at ($(n1) + (3em,0em)$) {$2$};%

        \node[vertex] (n3) at ($(n2) + (0em,3em)$) {$3$};%
        \node[vertex] (n4) at ($(n2) + (3em,0em)$) {$4$};
        \node[vertex] (n5) at ($(n4) + (0em,3em)$) {$5$};
        \node[vertex] (n6) at ($(n4) + (0em,-3em)$) {$6$};
        \node[vertex] (n7) at ($(n4) + (3em,0em)$) {$7$};
        \node[vertex] (n8) at ($(n7) + (0em,3em)$) {$8$};
        \node[vertex] (n9) at ($(n7) + (0em,-3em)$) {$9$};
        \node[vertex] (n10) at ($(n7) + (3em,0em)$) {$10$};
            \node[vertex] (n11) at ($(n10) + (0em,3em)$) {$11$};
        \node[vertex] (n12) at ($(n10) + (0em,-3em)$) {$12$};
        
	\draw[->,thick] (n1)--(n2);  
  \draw[->,thick] (n2)--(n3);  \draw[->,thick] (n2)--(n4);
  \draw[->,thick] (n4)--(n5);  \draw[->,thick] (n4)--(n6);
   \draw[->,thick] (n4)--(n7);
  \draw[->,thick] (n7)--(n8);  \draw[->,thick] (n7)--(n9);
 \draw[->,thick] (n7)--(n10);  \draw[->,thick] (n5)--(n4);
 \draw[->,thick] (n8)--(n11);  \draw[->,thick] (n9)--(n12);
	\draw[->,thick] (n2)--(n1);  
  \draw[->,thick] (n3)--(n2);  \draw[->,thick] (n4)--(n2);
  \draw[->,thick] (n5)--(n4);  \draw[->,thick] (n6)--(n4);
   \draw[->,thick] (n7)--(n4);
  \draw[->,thick] (n8)--(n7);  \draw[->,thick] (n9)--(n7);
 \draw[->,thick] (n10)--(n7);  \draw[->,thick] (n4)--(n5);
 \draw[->,thick] (n11)--(n8);  \draw[->,thick] (n12)--(n9);
\end{tikzpicture}
     \end{subfigure}}
 \subcaptionbox{
  \label{fig:gentop} 
    Generative Topology $G^T$} 
 {\begin{subfigure}{0.49\columnwidth}
\centering
 \begin{tikzpicture}[scale=0.65]
 \tikzstyle{vertex}=[circle,fill=none,minimum size=10pt,inner sep=0pt,thick,draw]
 \tikzstyle{pvertex}=[star,star points=14,fill=white,minimum size=10pt,inner sep=0pt,thick,draw]
  \node[vertex] (n1) {$1$};

  \node[vertex] (n2) at ($(n1) + (3em,0em)$) {$2$};%

        \node[vertex] (n3) at ($(n2) + (0em,3em)$) {$3$};%
        \node[vertex] (n4) at ($(n2) + (3em,0em)$) {$4$};
        \node[vertex] (n5) at ($(n4) + (0em,3em)$) {$5$};
        \node[vertex] (n6) at ($(n4) + (0em,-3em)$) {$6$};
        \node[vertex] (n7) at ($(n4) + (3em,0em)$) {$7$};
        \node[vertex] (n8) at ($(n7) + (0em,3em)$) {$8$};
        \node[vertex] (n9) at ($(n7) + (0em,-3em)$) {$9$};
        \node[vertex] (n10) at ($(n7) + (3em,0em)$) {$10$};
            \node[vertex] (n11) at ($(n10) + (0em,3em)$) {$11$};
        \node[vertex] (n12) at ($(n10) + (0em,-3em)$) {$12$};
 \draw[-,thick] (n1)--(n2);  
  \draw[-,thick] (n2)--(n3);  \draw[-,thick] (n2)--(n4);
  \draw[-,thick] (n4)--(n5);  \draw[-,thick] (n4)--(n6);
   \draw[-,thick] (n4)--(n7);
  \draw[-,thick] (n7)--(n8);  \draw[-,thick] (n7)--(n9);
 \draw[-,thick] (n7)--(n10);  \draw[-,thick] (n5)--(n4);
 \draw[-,thick] (n8)--(n11);  \draw[-,thick] (n9)--(n12);
 
\end{tikzpicture}
     \end{subfigure}}
 
    \caption{\label{fig:tree}A generative graph and its tree topology. }
    \end{figure}
\subsection{Graphical Representation}\label{sec:graph}
The structural description of  \eqref{eq:dim} induces a {\it generative graph} $G=(V,\overrightarrow{A})$ formed by identifying the set of vertices, $V=\{1,2,\dots ,n\},$ with random processes $x_i$ and the set of directed links, $\overrightarrow{A},$ is given by: $\overrightarrow{A}=\{i\to j | \trf _{ji}\ne 0\}$. Since we consider bi-directional dynamical systems, it follows that we have that a directed edge from $j$ to $i$ as well. Thus, $G$ is a bi-directed graph. Given generative graph $G$, its \emph{ generative topology} is the undirected graph  $G^T=(V,A)$ where $A=\{i-j \mid i\to j \in \overrightarrow{A}\} \cup \{i-j \mid i\gets j \in \overrightarrow{A} \}$. The following definitions on undirected graphs will be useful for subsequent analysis. Figure ~\ref{fig:tree} represents a bidirected system.
\begin{definition}[Path]Nodes $w_1,w_2, \dots , w_n\in V$ forms a path in an undirected graph, $G=(V,A)$, if for every $i=1,2,\dots, n-1$ we have $w_i-w_{i+1}$ in $A$. The path is denoted by $w_1-w_2\dots -w_n$. The length of the path is one less than the number of nodes in the path.   
\end{definition}
\begin{definition}[$n$-Hop Neighbor]
Given an undirected graph, $G=(V,A)$, a node $j\in V$ is a $n$ hop neighbor of $i\in V$ if there is a path of length $n$ between $i$ and $j$ in $G$. We will denote $n$ hop neighbors of $i$ as $n-hop(i)$. We refer 1-hop neighbors as neighbors.
\end{definition}

\begin{definition}[Tree]
An undirected graph $G=(V,A)$ is called a \emph{tree} if there is a unique path connecting any two nodes in $V$. 
\end{definition}
\begin{definition}[Leaf Node/ non-leaf nodes] In a tree, $G^T=(V,A)$, a node $i\in V$ that has only one neighbor ($1-hop(i)$) is called a leaf node. Nodes with more than one neighbor are called non-leaf nodes.  
\end{definition}
\noindent In figure ~\ref{fig:gentop}), $\{2,4,7,8,9\}$ are non-leaf nodes while the rest are leaf nodes. 
\begin{definition}[Radial Systems]
If the generative topology $G^T$ associated with a generative model $(\trf (z),e)$ is a tree, then the generative system is called a \emph{radial system}.
\end{definition}
\noindent Figure ~\ref{fig:tree} represents a radial system.
\begin{definition}[Kins]
Suppose the generative graph is $G=(V,A)$. The \emph{kins} of a node $i$, $kin(i)$, is given by: $kin(i)=\{ j|j\to i\text{ or } i\to j \text{ or } i\to k\gets j \text{ holds in } G\}$. 
\end{definition}
\noindent In figure ~\ref{fig:gengraph}), for example, $kin(2)=\{1,3,4,5,6,7\}$.
\begin{definition}[Moral-Graph]
\label{def:moralGraph}
Suppose $G=(V,\overrightarrow{A})$ is a generative graph. Its moral-graph is the undirected graph $G^M=(V,A^M)$ where $A^M:=\{i-j|j\in V, i\in kin(j)\}.$ 
\end{definition}
The moral graph for a radial system is the graph formed by adding undirected edges between 2-hop neighbors in the generative topology. See Figure ~\ref{fig:truemoral}) for example. 
\subsection{Moral Graph Inference from Time Series}
The relationship between the sparsity pattern of inverse PSD matrix of time sereis, $x$ and the moral graph, $G^M$ of a is described by the following result from \cite{MatSal12}.
\begin{theorem}
\label{thm:KinshipInf}
{\it
Consider a generative model $(\trf(z),e)$ consisting of N nodes with generative graph $G$. Let $x=(x_1,\dots,x_N)^T$ denote the time series measurements. Let $\Phi_{xx}$ be the power spectral density matrix  of the vector process $x$. Then the $(j,i)$ entry of $\Phi_{xx}^{-1}$ is non zero implies that $i$ is a kin of $j$.}
\end{theorem}

The basis of the above result comes from the structure of the matrix $\trf(z)$. For a radial system, we have $\trf_{ij}(z)\ne 0$ if and only if $i-j$ holds in $G^T$. From \eqref{eq:dimCompact} we can express $\Phi_{xx}$ for a radial system as follows:
\begin{equation}
\Phi_{xx} = (I-\trf(z))^{-1}\Phi_e(I-\trf(z))^{-*}.    
\end{equation}
The inverse PSD, $\Phi_{xx}^{-1}$, is given by:
\begin{align}
\Phi_{xx}^{-1} = (I-\trf(z))^{*}\Phi_e^{-1}(I-\trf(z))
\end{align}
We have the following:
\begin{equation}\label{eq:phiXinv}
\Phi^{-1}_{xx}(i,j)=\begin{cases}
-\trf_{ij}(z)\Phi ^{-1}_{e_i}-\trf_{ji}(z^{-1})\Phi ^{-1}_{e_{j}}, j \in 1-hop(i)\\
\trf_{ki}(z^{-1})\trf_{kj}(z)\Phi ^{-1}_{e_{k}}, j\in 2-hop(i) \\
\qquad \mbox{and } k\in 1-hop(i), k\in 1-hop(j)\\
\Phi ^{-1}_{e_i}+\sum _{k\in 1-hop(i)}|\trf_{ki}|^2\Phi ^{-1}_{e_k}, i=j\\
0, \mbox{otherwise.}
\end{cases}
\end{equation}
\begin{remark}\label{rem:otherdirectionMatSal} For $i$ and $j$ being kins but $\Phi_{xx}^{-1}(i,j)$ to be zero, the transfer functions in $\trf $ must be  belong to a set of measure zero on space of system parameters. For example,  system dynamics with transfer functions being zero or a static system with all noise sequences being identical. Therefore, except for these restrictive cases, the result in Theorem ~\ref{thm:KinshipInf}  is both necessary and sufficient. See \cite{MatSal12} for more details. For subsequent discussion and results to follow, we assume such pathological cases don't hold. 
\end{remark}
\section{Uncertainty Description}\label{sec:uncertain}
In this section we provide a description for how uncertainty affects the time-series $x_i.$ We interchangeably use corruption or perturbation to denote imperfections/uncertainties in measurement information.
\subsection{General Perturbation Models}\label{sec:Models}
Consider $i^{th}$ node in a generative graph and it's associated unperturbed time-series $x_i$. The corrupt data-stream $u_i$ associated with $i$ is assumed to follow:
\begin{equation}\label{eq:corruptionModel}
u_i[t]=g_i(x_i[\cdot],u_i[\cdot],\zeta_i[t ]),
\end{equation}
where  $u_i$ can depend dynamically (can be non-causal) on $x_i$, its own values in the strict past, and $\zeta_i[t]$ which represents
a stochastic process.
We highlight a few important perturbation models that are practically relevant. See \cite{SLS17network} for more details.
\subsubsection*{Temporal Uncertainty}
Consider a node $i$ in a generative graph. Suppose $t$ is the true clock index but the node $i$ measures a noisy clock index which is given by a random process, $\zeta _i[t]$. One such probabilistic model is given by the following Bernoulli process: 
\begin{equation*}
\zeta _i[t]=\begin{cases}
t_1, & \textrm{ with probability } p_i \\
t_2, & \textrm{ with probability } (1-p_i),
\end{cases}
\end{equation*}
where $t_1$ and $t_2$ are integers such that at least one of $t_1$ and $t_2$ are not equal to $0$. 
Randomized delays in information transmission can be modeled as a convolution operation with the impulse function $\delta [t]$ shifted by $\zeta _i[t]$ as follows :
\begin{equation}
\label{eq:randDelayMdl}
u_i[t] =\delta [t+\zeta _i[t]]*x_i[t].
\end{equation} 
\subsubsection*{Noisy Filtering}
Given a node $i$ in a generative graph, the data-stream $x_i$ is filtered by a stable filter $L_i$ and corrupted
with independent measurement noise $\zeta _i[\cdot ]$. This perturbation model is described by:
\begin{equation}
u_i[t] = (L_i * x_i)[t] + \zeta _i[t].
\end{equation}
\subsubsection*{Packet Drops}
The measurement $u_i[t]$ corresponding to a ideal measurement $x_i[t]$ packet reception at time $t$ can be stochastically modeled as: 
\begin{equation}
  \label{eq:packetDrop}	
u_i[t]=\begin{cases}
x_i[t], & \textrm{ with probability } p_i\\
u_i[t-1], & \textrm{ with probability } (1-p_i).
\end{cases}
\end{equation}
Consider a Bernoulli process $\zeta_i$ described by,
\begin{equation*}
\zeta _i[t]=\begin{cases}
1, & \textrm{ with probability } p_i \\
0, & \textrm{ with probability } (1-p_i).
\end{cases}
\end{equation*}
The corruption model in \eqref{eq:corruptionModel} takes the form: 
\begin{equation}
u_i[t]=\zeta _i[t]x_i[t]+(1-\zeta _i[t])u_i[t-1].
\end{equation}
\subsection{Corruption of power spectra}\label{subsec:corrupt_psd}
In all the perturbation models illustrated above, $u_i$ will have cross-spectra and power spectra of the form:
\begin{subequations}
\label{eq:spectrumConditions}
\begin{align}
  \Phi_{u_i x_i}(z) &= h_i(z) \Phi_{x_ix_i}(z) \\
  \Phi_{u_i u_i}(z) &= h_i(z)h_i(z^{-1}) \Phi_{x_i x_i}(z) + d_i(z),
\end{align}
\end{subequations}
for some transfer functions $h_i$ and $d_i$. If the perturbations were
deterministic and time invariant so that $u_i = h_i(z) x_i$, then the
power spectrum formulas would hold with $d_i(z)=0$. However, the
randomized perturbations imply that $d_i(z)\ne 0$. A more rigorous characterization of the perturbation models is described in  \cite{SLSarXivlinear}. 
\subsection*{Structure of inverse power spectra due to corruption}
Here, we will describe the structure of $\Phi_{uu}^{-1}(\omega )$. We will use the following equations and setup for deriving subsequent results. For compact notation, we will often drop the $\omega$ arguments.
  
  For $p=1,\ldots,N$ if $p$ is
  not a perturbed node, set $h_p(\omega) =1$ and $d_p(\omega) = 0$. With this notation,
  \eqref{eq:spectrumConditions} implies that the entries of $\Phi_{uu}$ are
  given by:
  \begin{equation*}
    (\Phi_{uu})_{pq} = \begin{cases}
      h_p  (\Phi_{xx})_{pq} h_q^{*}& \textrm{if } p\ne q \\
      h_p(\Phi_{xx})_{pp} h_p^{*}+ d_p & \textrm{if } p = q
    \end{cases}
  \end{equation*}
When $p\ne
  q$, there is no $d$ term because the perturbations were assumed to
  be independent.

  In matrix notation, we have that:
  \begin{equation*}
    \Phi_{uu} = H \Phi_{xx} H^{*} + \sum_{k=1}^n D_{v_k}
  \end{equation*}
  where $H$ is the diagonal matrix with entries $h_p$ on the
  diagonal and $D_{v_k}(\omega) = b_{v_k} d_{v_{k}}(\omega) b_{v_{k}}^T$ where
  $b_{v_{k}}$ is the canonical unit vector with $1$ at entry
  $v_{k}$. 

  For $k=0,\ldots,n-1$ set $\Psi_k = H \Phi_{xx} H^{*} + \sum_{m=1}^k
  D_{v_m}$. Here, $\Phi _{uu}^{-1}=\Psi ^{-1}_n$. We can inductively define these matrices as:
  \begin{subequations}
  \begin{align}
    \Psi_0 &= H \Phi_{xx}H^{*} \\
    \label{eq:inductiveMatrices}
    \Psi_{k+1} &= \Psi_k + b_{v_{k+1}}d_{v_{k+1}} b_{v_{k+1}}^T
  \end{align}
\end{subequations}
Note that $\Psi _0^{-1}(i,j)$ can be expressed as follows:
\begin{equation}
\label{eq:psi0}
    \Psi ^{-1}_0(i,j)=
    \frac{1}{\overline {h_i}(\omega)}\Phi _{xx}^{-1}(i,j)\frac{1}{h_j(\omega)}
\end{equation}

Combining Woodbury matrix identity in ~\eqref{eq:inductiveMatrices} implies that
\begin{equation}\label{eq:invPsiInduction}
  \Psi ^{-1}_{k+1}=\Psi ^{-1}_{k}-\Gamma _{k+1}
 \end{equation}
where $\Gamma_{k+1}:= \Psi^{-1} _{k}b_{v_{k+1}}b^T_{v_{k+1}}\Psi
^{-1}_{k}\Delta ^{-1}_{k+1},$ and $\Delta_{k+1}=d_{v_{k+1}}^{-1}+\Psi ^{-1}_{k}(v_{k+1},v_{k+1})$ which is a scalar. 

\subsection{Network identification in presence of corruption}
Here, we describe how structure learning using sparsity in inverse PSD of corrupted data streams leads to inference of spurious links. 
\begin{definition}[Perturbed Graph]
\label{def:corruptKG}
Let $G^M=(V,A^M)$ be a moral graph. Suppose $Y\subset V$ is
the set of corrupt nodes satisfying ~\eqref{eq:spectrumConditions}. Then the perturbed graph of
$G^M$ with respect to set $Y$ is the graph
$G ^U=(V,A^U)$ such that $i-j \in A^U$ if either
$i-j\in A^M$ or there is a path from  $i$ to $j$ in $G ^M$
such that all intermediate nodes are in $Y$. 
\end{definition}
We have the following result from \cite{SLS17network}.
\begin{theorem}\label{thm:multiperturbation}
\it{
Consider a generative model $(\trf(z),e)$ consisting of $N$ nodes with the moral graph
$G^M=(V,A^M)$. Let $\{v_1,v_2,\dots,v_n\}$ be the set of
$n$ perturbed nodes where each perturbation satisfies
\eqref{eq:spectrumConditions}.  
Then, 
$(\Phi_{uu}^{-1}(z))_{pq}\ne 0$ implies that $p$ and $q$ are neighbors in
the perturbed graph $G^U$. 
}
\end{theorem}
Consider a chain network consisting of 7 nodes with bidirectional dynamics between adjacent nodes as shown in ~\ref{fig:simple Chain}. The true moral graph is depicted in figure ~\ref{fig:truemoral}. Suppose $4$ is corrupted. Applying, Theorem \ref{thm:multiperturbation}, the inferred undirected graph is shown in figure ~\ref{fig:perturb 4}. 
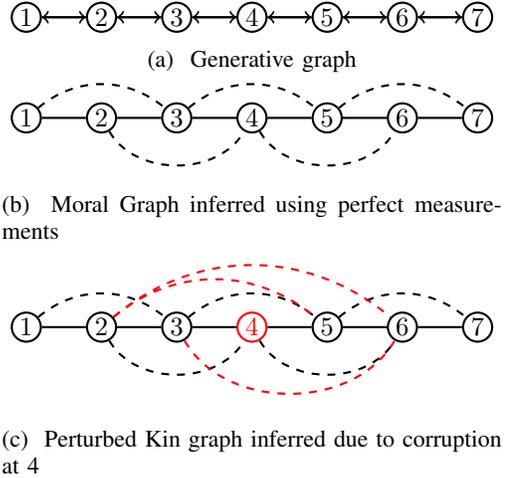
\begin{figure}[t]
  \centering
     \begin{subfigure}{0.75\columnwidth}
     \centering
         \begin{tikzpicture}
                 \tikzstyle{vertex}=[circle,fill=none,minimum size=11pt,inner sep=0pt,thick,draw]
         \tikzstyle{pvertex}=[color=red,circle,fill=white,minimum size=11pt,inner sep=0pt,thick,draw]
           \node[vertex] (n1) {$1$};
           \node[vertex, right of=n1] (n2) {$2$};
           \node[vertex,right of=n2] (n3) {$3$};
           \node[vertex,right of=n3] (n4) {$4$};
           \node[vertex,right of=n4] (n5) {$5$};
           \node[vertex,right of=n5] (n6) {$6$};
           \node[vertex,right of=n6] (n7) {$7$};

           \draw[->,thick] (n1)--(n2);
           \draw[->,thick] (n2)--(n3);
           \draw[->,thick] (n3)--(n4);
           \draw[->,thick] (n4)--(n5);
           \draw[->,thick] (n5)--(n6);
           \draw[->,thick] (n6)--(n7);
                      \draw[->,thick] (n2)--(n1);
           \draw[->,thick] (n3)--(n2);
           \draw[->,thick] (n4)--(n3);
           \draw[->,thick] (n5)--(n4);
           \draw[->,thick] (n6)--(n5);
           \draw[->,thick] (n7)--(n6);
           
         \end{tikzpicture}
         \subcaption{\label{fig:simple Chain} Generative graph}
   \end{subfigure}
     \begin{subfigure}{0.75\columnwidth}
     \centering
         \begin{tikzpicture}
                 \tikzstyle{vertex}=[circle,fill=none,minimum size=11pt,inner sep=0pt,thick,draw]
         \tikzstyle{pvertex}=[color=red,circle,fill=white,minimum size=11pt,inner sep=0pt,thick,draw]
           \node[vertex] (n1) {$1$};
           \node[vertex, right of=n1] (n2) {$2$};
           \node[vertex,right of=n2] (n3) {$3$};
           \node[vertex,right of=n3] (n4) {$4$};
           \node[vertex,right of=n4] (n5) {$5$};
           \node[vertex,right of=n5] (n6) {$6$};
           \node[vertex,right of=n6] (n7) {$7$};

           \draw[thick] (n1)--(n2);
           \draw[thick] (n2)--(n3);
           \draw[thick] (n3)--(n4);
           \draw[thick] (n4)--(n5);
           \draw[thick] (n5)--(n6);
           \draw[thick] (n6)--(n7);
           \draw[thick,dashed] (n1) to[out=40,in=140] (n3);
           \draw[thick,dashed] (n3) to[out=40,in=140] (n5);
           \draw[thick,dashed] (n5) to[out=40,in=140] (n7);
           \draw[thick,dashed] (n2) to[out=300,in=240] (n4);
           \draw[ thick,dashed] (n4) to[out=300,in=240] (n6);
         \end{tikzpicture}
         \subcaption{\label{fig:truemoral} Moral Graph inferred using perfect measurements}
   \end{subfigure}
     \begin{subfigure}{0.75\columnwidth}
     \centering
         \begin{tikzpicture}
                 \tikzstyle{vertex}=[circle,fill=none,minimum size=11pt,inner sep=0pt,thick,draw]
         \tikzstyle{pvertex}=[color=red,circle,fill=white,minimum size=11pt,inner sep=0pt,thick,draw]
           \node[vertex] (n1) {$1$};
           \node[vertex, right of=n1] (n2) {$2$};
           \node[vertex,right of=n2] (n3) {$3$};
           \node[pvertex,right of=n3] (n4) {$4$};
           \node[vertex,right of=n4] (n5) {$5$};
           \node[vertex,right of=n5] (n6) {$6$};
           \node[vertex,right of=n6] (n7) {$7$};

          \draw[thick] (n1)--(n2);
           \draw[thick] (n2)--(n3);
           \draw[thick] (n3)--(n4);
           \draw[thick] (n4)--(n5);
           \draw[thick] (n5)--(n6);
           \draw[thick] (n6)--(n7);
           \draw[thick,dashed] (n1) to[out=40,in=140] (n3);
           \draw[thick,dashed] (n3) to[out=40,in=140] (n5);
           \draw[thick,dashed] (n5) to[out=40,in=140] (n7);
           \draw[thick,dashed] (n2) to[out=300,in=240] (n4);
           \draw[ thick,dashed] (n4) to[out=300,in=240] (n6);
           \draw[red, thick,dashed] (n2) to[out=40,in=140] (n5);
           \draw[red,thick,dashed] (n2) to[out=40,in=140] (n6);
           \draw[red, thick,dashed] (n3) to[out=300,in=240] (n6);
         \end{tikzpicture}
         \subcaption{
         \label{fig:perturb 4} Perturbed Kin graph inferred due to corruption at 4
         }
   \end{subfigure}
   \caption{
    \label{fig:corrupt spectra} This figure shows how unreliable measurements at a node can yield in erroneous dynamic influences. 
  }
\end{figure}
\section{Exact Topology Learning}
The first step towards exact topology learning is to determine the location of all the corrupt nodes. This is presented in the following subsection.
We consider the following assumption on the location of corrupt nodes:
\begin{assumption}\label{assume:Zloc}
\begin{enumerate}[C1)]
    \item \label{case:leaf} Corrupt nodes are at least 3 hops away from all leaf nodes in the generative topology.
    \item \label{case:corrupt}Corrupt nodes are at least 3 hops away from each other in the generative topology.
    \end{enumerate}
       \end{assumption}
   \begin{remark}
   The above condition C\ref{case:leaf}) implies that the corrupt nodes are located deep in the network such that its effects are felt by the agents that have perfect measurements. 
   \end{remark}  
\subsection{Corruption Detection} \label{sec:corruptionDetection}In this section we describe a method to locate the corrupt nodes in the inferred perturbed graph for radial dynamical systems.
\subsection*{Neighborhood characterization}\label{sec:clique}
In this subsection, we characterize the neighborhood set of leaf and corrupt nodes. The following proposition will be needed for the development to follow. 
   \begin{proposition}\label{prop:nbru}
\it{Consider a radial system with generative topology $G^T=(V,A)$ consisting of $N$ nodes with the moral graph
$G^M=(V,A^M)$. Let $1-hop (i)$ and $2-hop(i)$ denote the set of 1-hop and 2 hop neighbors of $i$ in $G^T$. Let $Y\subset V $ be the set of perturbed nodes where each perturbation satisfies
\eqref{eq:spectrumConditions} and Assumption \ref{assume:Zloc}.
Suppose $G ^U=(V,A^U)$ is the perturbed graph inferred using Theorem \ref{thm:multiperturbation}. Let neighbors of node $i$ in $G^U$ ben $\nbr_u(i)$. If $i$ is a leaf node in $G^T$ or $i\in Y$, then $\nbr _u(i)=1-hop (i)\cup 2-hop(i)$.}
   \end{proposition}
\begin{proof}
We will show that no additional nodes excluding $1-hop(i)$ or $2-hop(i)$ neighbors exist as neighbors of $i$ in $G^U$. Suppose $j$ is a neighbor of $i$ in $G^U$ such that $j\notin 1-hop(i)\cup 2-hop(i)$. Then, by definition of perturbed graph, there should be a path $i-v_1-v_2-j$ in $G^M$ such that $v_1$ and $v_2$ are corrupt nodes. This implies $v_1$ belongs to $1-hop(i)$ or $2-hop(i)$ in $G^T$. Suppose $i$ is a corrupt node.  This contradicts condition C\ref{case:corrupt}). Suppose $i$ is a leaf node. This contradicts condition C\ref{case:leaf}).Therefore, $\nbr _u(i)=1-hop (i)\cup 2-hop(i)$.
%
%
%
\end{proof}
The following lemma describes a topological method to detect a set of candidate nodes which contains only leaf and corrupt nodes using the perturbed graph. It states that only leaf nodes and corrupt nodes has a neighborhood that forms a clique in the perturbed graph.
\begin{lemma}\label{lemma:leafCorrupt}
\it{Consider a radial dynamical system with generative topology $G^T=(V,A^T).$ Let $Y\subset V$ be the set of perturbed nodes where each perturbation satisfies
\eqref{eq:spectrumConditions} and assumption \ref{assume:Zloc}.
Suppose $G ^U=(V,A^U)$ is the perturbed graph inferred using theorem \ref{thm:multiperturbation}. Consider any node $i$ in $V$. Neighbors of node $i$ in $G^U$, $\mathcal{N}_u(i)\cup \{i\}$ will form a clique in $G ^U$ if and only if $i$ is a leaf node in generative topology $G ^T$ or $i$ is a corrupt node.}
\end{lemma}
\begin{proof}
$(\Rightarrow)$ We will show that if $i$ is neither a leaf node nor a corrupt node, then $\nbr_u(i)\cup \{i\}$ does not form a clique in $G^U$. Note that all nodes in $1-hop (i)$ and all nodes in $2-hop (i)$ will be neighbors of $i$ in moral graph and hence are neighbors of $i$ in $G^U$. We will show that there exists a pair of nodes $a,b\in \nbr_u(i)$ such that $a-b$ does not hold true in $G^U$ and thus $\nbr(i)_u\cup \{i\}$ cannot form a clique in $G^U$.

By generative model description, we have that there is at least one 1-hop neighbor and one 2 -hop neighbor for every node in $G^T$. Consider $a \in 1-hop(i)$ and $b\in 2-hop(i)$. Then, a path $a-i-p-b$ exists in $G^T$ for some node $p\in V$. Either $a$ is a leaf node or not in $G^T$. Suppose, $a$ is a leaf node. As $G^T$ is a tree, the path $a-i-p-b$ is unique.
Thus, all possible paths between $a$ and $b$ in $G^M$ goes through at least one of $i,p$. As $a$ is a leaf node, by condition C~\ref{case:leaf}), $i,p$ are not corrupt. Thus, $a-b \notin A^U$.

Suppose $a$ is not a leaf node. Then, there exists a path $c-a-i-p-b$ in $G^T$. Thus, $c\in 2-hop(i)$ and a neighbor of $i$ in the perturbed graph, $G^U$. We show that $c,b$ despite being neighbors of $i$ in $G^U$, $c-b$ does not hold true in $G^U$. Similar to argument above, since $G^T$ is a tree, the path connecting $c,b$, $c-a-i-p-b$ is unique in $G^T$. Thus, all possible paths between $c$ and $b$ in $G^M$ goes through at least one of $a$, $p$ and $i$. By condition C\ref{case:corrupt}), both $a$ and $p$ cannot be corrupt. Thus all possible paths between $c$ and $b$ in $G^M$ goes through at least one non-perturbed node. Thus, $c-b\notin A^U$.        

$(\Leftarrow)$ Suppose $i$ is a leaf node in $G^T$ or a corrupt node. We will show that $\nbr_u (i)\cup \{i\}$ forms a clique in the perturbed kin graph, $G^U$. Using Proposition ~\ref{prop:nbru}, $\nbr_u(i)\cup\{i\}=1-hop(i)\cup 2-hop(i)\cup \{i\}$. 

\emph{$i$ is a leaf node:} There is only one non-leaf node $nl$ which is a neighbor of $i$ in $G^T$.  
Any pair of 2-hop neighbors of $i$ in $G^T$, $k_1,k_2$, has a common parent of $nl$ in the generative graph. Thus, $k_1-nl-k_2$ holds in $G^T$ and neighbors in $G^M$. Hence, $k_1-k_2$ holds in $G^U$. Thus, $\nbr_u(i)\cup \{i\}$ forms a clique in $G^U$. 

\emph{$i$ is a corrupt node:} For any $k_1$, $k_2\in \nbr _u(i)$, there is a path $k_1-i-k_2$ in the moral graph $G^M$. As $i$ is a corrupt node, $k_1-k_2$ holds in $G^U$. Thus, $\nbr_u(i)\cup \{i\}$ forms a clique in $G^U$.
%
\end{proof}
\subsection*{Detection of corrupt nodes}\label{sec:detectY}
After a candidate set containing corrupt nodes and leaf nodes are determined as discussed above, we will now isolate the corrupt nodes exactly. The following theorem precisely detects the corrupt nodes separately based on phase properties of entries in inverse PSD.    
\begin{theorem}\label{thm:phase}
\it{ Suppose $G ^T=(V,A^T)$ is the generative topology corresponding to a radial dynamical system. Let $Y=\{v_1,v_2,\dots ,v_n\}\subset V$ be set of corrupt nodes with each corruption satisfying ~\eqref{eq:spectrumConditions} and  Assumption ~\ref{assume:Zloc}. Suppose $G ^U=(V,A^U)$ is the perturbed graph inferred using theorem \ref{thm:multiperturbation}. Let $B$ be the set of nodes detected using Lemma \ref{lemma:leafCorrupt} whose neighborhood, $\nbr _u(i)$ forms a clique with $\{i\}$ in $G^U$. Take a node $i \in B.$  Then, $i$ has at least two neighbors, $p$, $q$ in $G^U$ with  non-constant $\angle \Phi _{uu}^{-1}(i,p)(\omega)$ for all $\omega \in (-\pi,\pi]$ if and only if $i$ is a corrupt node in $G ^T$.}
\end{theorem}
\begin{proof}
 First, we recall the structure of $\Phi_{uu}^{-1}(\omega )$ described in \ref{subsec:corrupt_psd} in equations \eqref{eq:inductiveMatrices}-\eqref{eq:invPsiInduction}.
For each $i \in B$, we will inductively prove that for $k=2,\dots , n$, $\Psi ^{-1}_k(i,j)=\Psi ^{-1}_1(i,j)$. 
Consider case $k=2$. Using ~\eqref{eq:invPsiInduction},   
$\Psi ^{-1} _{2}(i,j)=\Psi ^{-1}_{1}(i,j)-\Gamma _{2}(i,j)$. Note that $\Gamma _2(i,j)=\Psi ^{-1}_1(i,v_2)\Psi ^{-1}_1(v_2,j)\Delta ^{-1}_{2}$. Similarly, using ~\eqref{eq:invPsiInduction}, $\Psi ^{-1} _{1}(i,v_2)=\Psi ^{-1}_{0}(i,v_2)-\Gamma _{1}(i,v_2)$ where $\Gamma_1(i,v_2):= \Psi^{-1} _{0}(i,v_1)\Psi
^{-1}_{0}(v_1,v_2)\Delta ^{-1}_{1}.$ Here, if $i$ is a corrupt node, $v_1 =i$.
 By ~\eqref{eq:psi0}, $\Psi ^{-1}_0(v_1,v_2)=\frac{\Phi ^{-1}_{xx}(v_1,v_2)}{\overline{h_{v_1}}h_{v_2}}$. 
As $v_1$ and $v_2$ are at least 3-hops away in $G^T$, using ~\eqref{eq:phiXinv}, $\Phi ^{-1}_{xx}(v_1,v_2)=0$. Thus, $\Gamma _1(i,v_2)=0.$ Again, by ~\eqref{eq:psi0} $\Psi _{0}^{-1}(i,v_2)=\frac{\Phi ^{-1}_{xx}(i,v_2)}{\overline{h_{i}}h_{v_2}}.$ Since $i$ is either a leaf node or a corrupt node, we have  that $i,v_2$ are at least 3 hops away from each other. Using this and ~\eqref{eq:phiXinv} we have that $\Phi ^{-1}_{xx}(i,v_2)=0$. This implies $\Psi _{0}^{-1}(i,v_2)=0$. Therefore, $\Psi _{1}^{-1}(i,v_2)=0$ and hence $\Gamma _2(i,j)=0$. Thus we have proved that $\Psi ^{-1} _{2}(i,j)=\Psi ^{-1}_{1}(i,j)$.  

Now assume that the claim holds for some $k>2$. That is, $\Psi ^{-1}_k(i,j)=\Psi ^{-1}_1(i,j).$ Using ~\eqref{eq:invPsiInduction},   
$\Psi ^{-1} _{k+1}(i,j)=\Psi ^{-1}_{k}(i,j)-\Gamma _{k+1}(i,j)$, where $\Gamma _{k+1}(i,j)=\Psi ^{-1}_{k}(i,v_{k+1})\Psi ^{-1}_k(v_{k+1},j)\Delta ^{-1}_{k+1}$. Using the induction hypothesis, $\Psi ^{-1}_k(i,v_{k+1})=\Psi ^{-1}_1(i,v_{k+1}).$ 
As $v_{1},v_{k+1}$ and $i,v_{k+1}$ are at least 3 hops away from each other respectively, using the same argument as described for $v_1,v_2$ and $i,v_2$ in the previous paragraph, we have $\Gamma_{k+1}(i,v_{k+1})=0$ and $\Psi ^{-1} _{k+1}(i,j)=\Psi ^{-1}_{k}(i,j)=\Psi ^{-1}_{1}(i,j)$. 
As $\Phi ^{-1}_{uu}=\Psi ^{-1}_{k}$ for $k = n$, we have established that $\Phi ^{-1}_{uu}(i,j)=\Psi ^{-1}_{1}(i,j)$. 

$(\Leftarrow)$ We will show that if $i\in B$ is a leaf node in $G^T$, then there is at most only one node $j\in \nbr _u(i)$ such that $\angle \Phi _{uu}^{-1}(i,j)(\omega)$ is non-constant for all $\omega \in (-\pi,\pi]$. By Proposition ~\ref{prop:nbru},  $\nbr_u(i)= 1-hop(i)\cup 2-hop(i)$. Moreover as $i$ is a leaf node, any node $j\in \nbr _u(i)$ is not a corrupt node. Therefore, using ~\eqref{eq:psi0} and preceding discussion, we have 
\begin{equation}\label{eq:leafUU}
    \Phi ^{-1}_{uu}(i,j) =  \Psi ^{-1} _{1}(i,j)=\Psi ^{-1}_{0}(i,j) =\Phi _{xx}^{-1}(i,j).
\end{equation}

Since $i$ is a leaf node, there is only node in $1-hop(i)$. Suppose $r$ is that node. We will show that for all $j\ne r\in \nbr _u(i)$ (this means $j\in 2-hop(i)$), $\angle \Phi ^{-1}_{uu}(i,j)=0$ while $\angle \Phi ^{-1}_{uu}(i,r)$ is non-constant.

Take any $j\in 2-hop(i)$. Let $q\in V$ be the common neighbor of $i$ and $j$ in $G^T$. Combining ~\eqref{eq:leafUU} and \eqref{eq:phiXinv} we have:
 \begin{align}
 \Phi _{uu}^{-1}(i,j)&= \overline{\mathcal{G}}_{qi}(\omega)\mathcal{G}_{qj}(\omega)\Phi ^{-1}_{e_q}\\
 &=\frac{b_{qi}b_{qj}\Phi ^{-1}_{w_q}}{|S_{q}|^4},     
 \end{align} which is a positive real scalar. Thus, $\angle \Phi _{uu}^{-1}(i,j)=0$. 
 
 Now, consider the node $r$.  By ~\eqref{eq:leafUU} and ~\eqref{eq:phiXinv} we have:
   \begin{equation}
       \Phi ^{-1}_{uu}(i,r)=-\frac{b_{ir}\Phi ^{-1}_{w_{i}}}{S_{i}(\omega)|S_{i}(\omega)|^2}-\frac{b_{ri}\Phi ^{-1}_{w_{r}}}{S_{r}(\omega)|S_{r}|^2}
   \end{equation}
Thus $\Phi ^{-1}_{uu}(i,r)$ has a non-constant phase response.

$(\Rightarrow)$ We will show that if $i\in B$ is a corrupt node in $G^T$, then there are at least two neighbors $p,r$ of $i$ in $G^U$ such that $\Phi ^{-1}_{uu}(i,p) $, $\Phi ^{-1}_{uu}(i,r)$ are non-constant transfer functions. By assumption on location of corrupt nodes, every corrupt node has at least two 1 hop neighbors in $G^T$. Take any $p\in 1-hop(i).$ We will now show that $\Phi _{uu}^{-1}(i,p)$ is not a constant transfer function.

Now, $\Phi ^{-1} _{uu}(i,p)=\Psi ^{-1}_{0}(i,p)-\Psi^{-1} _{0}(i,i)\Psi
^{-1}_{0}(i,p)\Delta ^{-1}_{1}.$ By ~\eqref{eq:psi0}, $\Psi ^{-1}_0(i,p)=\frac{\Phi ^{-1}_{xx}(i,p)}{\overline{h_i}h_{p}}$. Then,
\begin{equation}\label{eq:expandPhiuu}
\Phi ^{-1} _{uu}(i,p)= \frac{\Phi _{xx}^{-1}(i,p)}{\overline{h_i}}-\frac{\Delta _1 ^{-1}\Phi ^{-1}_{xx}(i,i)\Phi ^{-1}_{xx}(i,p)}{\overline{h_i}|h_{i}|^2}   
\end{equation}
 Using ~\eqref{eq:phiXinv} we have:
   \begin{equation}\label{eq:expandphiuu2}
       \Phi ^{-1}_{xx}(i,p)=-\frac{b_{ip}\Phi ^{-1}_{w_{i}}}{S_{i}(\omega)|S_{i}(\omega)|^2}-\frac{b_{pi}\Phi ^{-1}_{w_{p}}}{S_{p}(\omega)|S_{p}|^2}
   \end{equation}
   Therefore, by ~\eqref{eq:expandPhiuu} and ~\eqref{eq:expandphiuu2} we can see that $\Phi _{uu}^{-1}(i,p)$ is not a constant transfer function and hence has non-constant phase response. 
\end{proof}

The above result detects the set of corrupt nodes, $y$, from the candidate set, $B$, and hence the remaining nodes $B\setminus y$ are the leaf nodes. Moreover, the above result delineates that only leaf nodes have one unique entry in $\Phi ^{-1}_{uu}$ with a non-constant phase response. This corresponds to the true edge associated with the leaf node. Thus, the above also provides a method to detect leaf nodes and remove spurious edges associated with leaf nodes. The procedure is described comprehensively in Algorithm ~\ref{alg:detectCorrupt}.
\begin{algorithm}
	\caption{\label{alg:detectCorrupt} Detection of Corrupt Nodes and Isolating True Edges Associated with Leaf Nodes}
	\begin{algorithmic}[1]
		\Statex \textbf{Input:} Time series measurements, $u$. 
		\Statex \textbf{Output:} Set of perturbed nodes, $Y$, set of leaf nodes, $L$, and set of true edges, $\mathcal{E}_L$ associated with leaf nodes. 
		\Statex \textbf{Init:} $A_Z\gets \{ \}$, $Y\gets \{ \}$, $L\gets \{ \}$, $\mathcal{E}_L \gets \{ \}$.
		\State Compute inverse PSD, $\Phi _{uu}^{-1}$.		
		\ForAll{$i \in V$, $i\neq j$}
			\If{$\Phi_{uu}^{-1}(i,j)(\omega )\ne 0$}
				\State $A_Z\gets A_Z \cup \{i-j\}$
			\EndIf
		\EndFor
		\ForAll{$i\in V$}
		\State $\epsilon _i \gets \{ \}$.
			\If{$\nbr _u(i)\cup \{i\}$ forms a clique in $G_Z=(V,A_Z)$}
              \ForAll{$j\in \nbr _u(i)$}
				    \If{ $\angle \Phi_{uu}^{-1}(i,j)(\omega )$ is not constant for all $\omega \in (-\pi,\pi]$}
					    \State $\epsilon _i \gets \epsilon _i \cup\{ i-j \}$
				    \EndIf
              \EndFor
          \If{Cardinality of $\epsilon _i\ge 2$}
              \State $Y\gets Y\cup \{i\}$
          \Else{$ L\gets L \cup \{i\}$ and $\mathcal{E}_{L}\gets \mathcal{E}_{L}\cup \epsilon _i$}
          \EndIf
		  \EndIf
	\EndFor
	\end{algorithmic}
\end{algorithm}
\subsection{Hide and Learn Algorithm}\label{sec:exactLearn}
\begin{algorithm}
	\caption{\label{alg:hidelearn} Exact Topology Learning: \emph{Hide and Learn}}
	\begin{algorithmic}[1]
		\Statex \textbf{Input:} Inputs and outputs from Algorithm ~\ref{alg:detectCorrupt}.
		\Statex \textbf{Output:} Set of true edges, $A$, in generative topology $G^T$.
		\Statex \textbf{Init:} Set of observed edges, $A_o\gets \{ \}$.
		\State Isolate non-corrupt measurements, $o=u\setminus y$. Observed nodes $V_o=V\setminus Y$. 
		\State Using measurements $o$, compute inverse PSD, $\Phi _{oo}^{-1}$.		
		\ForAll{$i \in V_o$, $i\neq j$}
			\If{$\Phi_{oo}^{-1}(i,j)(\omega )\ne 0$}
				\State $A_o\gets A_o \cup \{i-j\}$
			\EndIf
		\EndFor	 
		\State Non-leaf nodes, $V_{nl}=V\setminus \{Y\cup L\}$.
		\State True edge set, $\mathcal{E}_T \gets \mathcal{E}_L$. 
		\ForAll{$p,q \in V_{nl}$ such that $p-q \in A_o$}
			\If{There exist $K \ne \{\}$ and $S\ne \{\}$ such that $sep(K,S|\{p,q\})$ holds}
				\State $\mathcal{E}_T \gets \mathcal{E}_T \cup \{p-q\}$
				\State $A\gets A\cup \mathcal{E}_T$
			\EndIf
		\EndFor
		\State $d \gets $ number of disconnected components in the graph, $\Theta =(V_o,\mathcal{E}_T)$. (i.e $\Theta =\cup _{i=1}^{d}\theta _i$).
		\ForAll {$i\in \{1,2,\dots,d\}$}
			\ForAll{$j\in \{i+1,\dots,d\}$}
				\If{There exists nodes $q\in \theta _i$ and $r\in \theta _j$ such that $p-q \in \theta _i$ and $s-r \in \theta _j$ holds for some other observed nodes $p,s$}
					\ForAll{ $l\in Y$}
					\If{ $\{p,q,l,r,s\}$ forms a clique in $G_Z$}
						\If{$\Phi _{uu}^{-1}(p,s)(\omega)$ is constant for all $\omega \in (-\pi,\pi]$}
							\State $A \gets A \cup \{q-l,l-r\}$
						\EndIf
					\EndIf
					\EndFor
				\EndIf
			\EndFor
		\EndFor
\end{algorithmic}
\end{algorithm} 

The steps to recover the exact topology of the radial linear dynamical system using imperfect information are presented in this section. To accomplish this we follow \textit{hide and learn} strategy. This is described  in Algorithm ~\ref{alg:hidelearn}. First, hide the measurements of the corrupt nodes. We  infer the graphical structure of the network by observing sparsity pattern of inverse PSD using only the nodes that has perfect information by marginalizing out the corrupt node measurements. That is, the corrupt nodes will be treated as latent nodes. This graph will contain spurious edges. This constitutes lines 1 to 7 in Algorithm ~\ref{alg:hidelearn}. Second, identify the true edges in the graph obtained from previous step. This constitutes lines 8 to 15 in Algorithm ~\ref{alg:hidelearn}. Finally, place the corrupt nodes back at the correct location in the structure resulting from previous step as described in lines 16 to 29 in Algorithm ~\ref{alg:hidelearn}. 
Theorem ~\ref{thm:mainalg} is the main result of the article which states that Algorithm ~\ref{alg:hidelearn} precisely learns the exact topology of a radial system with imperfect information once the corrupt nodes have been detected using Algorithm ~\ref{alg:detectCorrupt}. The proof is given in the appendix. 
\begin{theorem}\label{thm:mainalg}
Suppose $Y$ is the set of perturbed nodes, $L$ is the set of leaf nodes and $\mathcal{E}_L$ is the set of true edges associated with leaf nodes detected from Algorithm ~\ref{alg:detectCorrupt}. Then, Algorithm ~\ref{alg:hidelearn} results in learning the true generative topology is $G^T=(V,A)$ for the corresponding radial system.
\end{theorem}
\begin{figure}[t]
 \centering
 \includegraphics[height=0.7\columnwidth,width=   0.9\columnwidth]{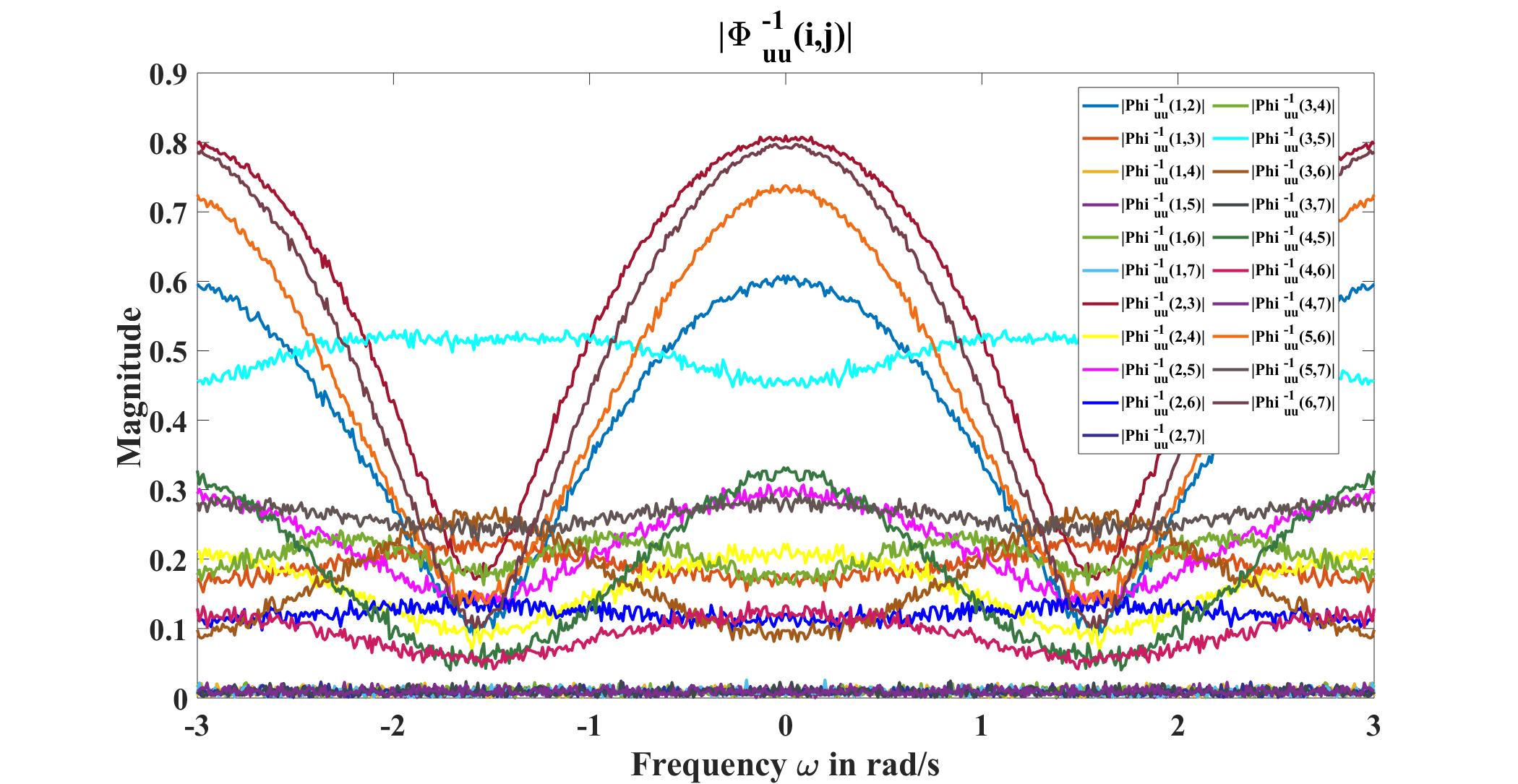}
 \caption{ \label{fig:magu} {\textbf{Magnitude Plots}} The magnitude of inverse power spectral density estimates computed from corrupt data streams $u$ are shown in the here. Notice the entries are non-zero across the frequency grid. To each non-zero entry, we add undirected edges to infer the perturbed graph as shown in Figure ~\ref{fig:perturb 4}) following theorem ~\ref{thm:multiperturbation}.}  
\end{figure}
\section{Simulation Result}\label{sec:sim}
In this section we demonstrate the topological learning algorithm via a numerical example. Let the true generative graph, $G$, be as shown in Fig.~\ref{fig:simple Chain}) with the following dynamics: 
  \begin{equation}\label{eq:eg}
\begin{aligned} 
x_1[t]&= 0.5 x_2[t-1] + e_1[t]\\
x_2[t]&= 0.36 x_1[t-1] +0.6 x_3[t-1]+ e_2[t]\\
x_3[t]&= 0.95 x_2[t-1] -1.7 x_4[t-1]+ e_3[t]\\
x_4[t]&= 0.51 x_3[t-1] +0.55 x_5[t-1]+ e_4[t]\\
x_5[t]&= 1.5 x_4[t-1] +0.6 x_6[t-1] + e_5[t]\\
x_6[t]&= 0.7 x_5[t-1] +0.5 x_7[t-1] + e_6[t]\\
x_7[t]&= 0.65 x_6[t-1]+  e_7[t]
\end{aligned}
\end{equation}
where $e_i$ are white noise sequences. The corruption model for node 4 is:
\begin{equation}
u_4[t] = \begin{cases}
  x_4[t-2], & \textrm{ with probability } 0.7 \\
  x_4[t], & \textrm{ with probability } 0.3.
\end{cases}
\end{equation} 
\begin{figure}
 \centering
 \includegraphics[height=0.7\columnwidth,width=0.9\columnwidth]{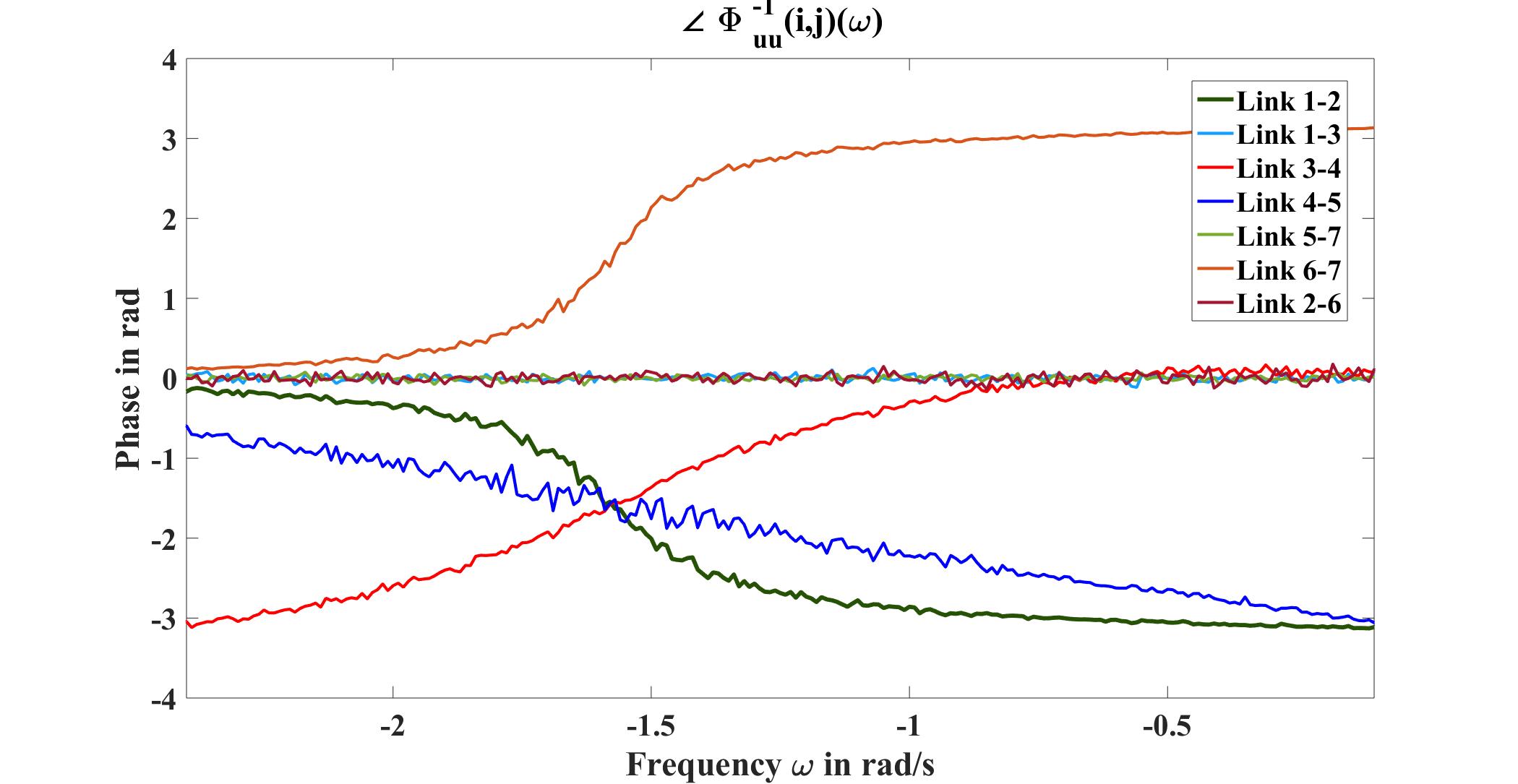}
 \caption{ \label{fig:phase} {\textbf{Phase Plots}.} The estimated phase response values are shown in the figure. We observe that the phase response corresponding to edges to the leaf nodes, $\{1,7\}$, have only non-constant phase response. Node $4$ has two non-constant phase response. This verifies predictions of Theorem ~\ref{thm:phase}. The phase response estimate of $2-6$ link is approximately close to zero and is a constant. This verifies Lemma ~\ref{lemma:placecorrupt}. } 
\end{figure}
From a trajectory length of $10^7$, the estimates for power spectral density was obtained using MATLAB 'cpsd' command. The plot for
magnitude of the inverse power spectral density estimates is shown in Figure ~\ref{fig:magu}. \emph{Step 1:} Using Theorem ~\ref{thm:multiperturbation}, adding edges and constructing an undirected graph results in the perturbed graph shown in figure ~\ref{fig:perturb 4}). \emph{Step 2:} We notice that neighbors of 1,4 and 7 forms a clique with nodes 1,4 and 7 respectively. As predicted by Lemma ~\ref{lemma:leafCorrupt}, we have identified the candidate set. \emph{Step 3:} The next step is to detect the corrupt node. To this we observe the phase response of the inverse PSD estimated. Figure ~\ref{fig:phase} shows that only $4$ will have at least two non-constant phase estimates. For leaf nodes, there will only be one non-constant phase plot. Using Theorem ~\ref{thm:phase} we determine node 4 as the corrupt node. 
\emph{Step 5:} The next step is to follow the \textit{hide and learn} algorithm. We first remove the measurements of node $4$ and infer the topology of the network with latent node 4. That is, using the measurements $O=\{1,2,3,5,6,7\}$, we compute the inverse PSD. The magnitude of $\Phi ^{-1}_{oo}$ is shown in Figure ~\ref{fig:latentplot}. Following Lemma ~\ref{lemma:latentPSD} yields the undirected graph shown in figure ~\ref{fig:latent4}.
\begin{figure}
 \centering
 \includegraphics[height=0.7\columnwidth,width=0.9\columnwidth]{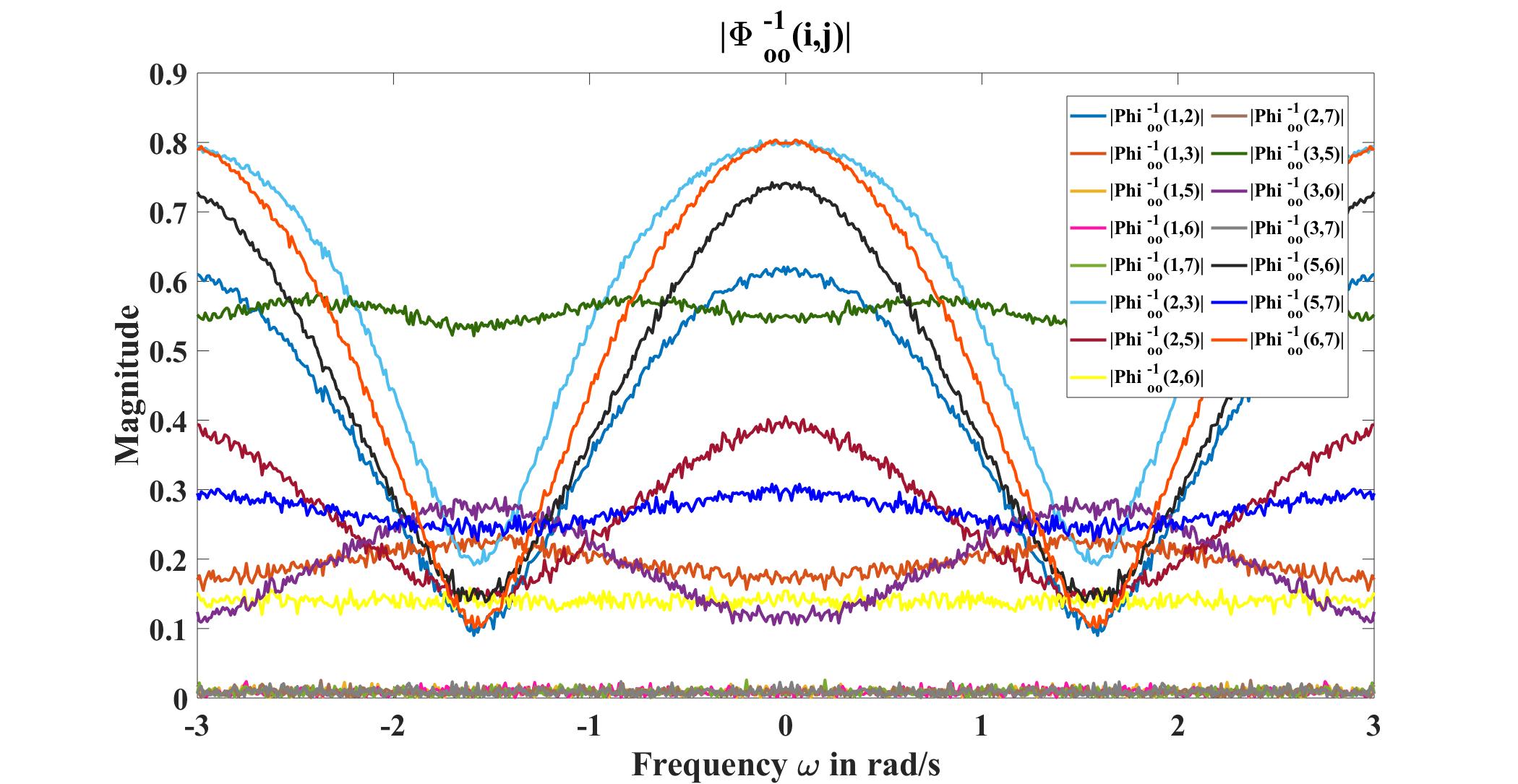}
 \caption{ \label{fig:latentplot} {\textbf{Unobserved node $4$}.} The magnitude of the inverse PSD estimates computed from $o=u\setminus \{4\}$ are shown here. $y$ axis is angular frequency $\omega$ in radians/s. Notice the entries are non-zero across the frequency grid. To each non-zero entry, we add undirected edges to infer the undirected graph as shown in Figure ~\ref{fig:latent4}) following Lemma ~\ref{lemma:latentPSD}. }
\end{figure}
\begin{figure}[t]
  \centering
  \begin{subfigure}{0.7\columnwidth}
        \centering
      \begin{tikzpicture}
                \tikzstyle{vertex}=[circle,fill=none,minimum size=11pt,inner sep=0pt,thick,draw]
        \tikzstyle{pvertex}=[color=red,circle,fill=white,minimum size=11pt,inner sep=0pt,thick,draw]
          \node[vertex] (n1) {$1$};
          \node[vertex, right of=n1] (n2) {$2$};
          \node[vertex,right of=n2] (n3) {$3$};
          \node[vertex,right of=n3] (n5) {$5$};
          \node[vertex,right of=n5] (n6) {$6$};
          \node[vertex,right of=n6] (n7) {$7$};

          \draw[thick] (n1)--(n2);
          \draw[thick] (n2)--(n3);
          \draw[thick] (n3)--(n5);
          \draw[thick] (n5)--(n6);
          \draw[thick] (n6)--(n7);
          \draw[thick,dashed] (n1) to[out=40,in=140] (n3);
          \draw[thick,dashed] (n5) to[out=40,in=140] (n7);
          \draw[red, thick,dashed] (n2) to[out=40,in=140] (n5);
          \draw[red,thick,dashed] (n2) to[out=40,in=140] (n6);
          \draw[red, thick,dashed] (n3) to[out=300,in=240] (n6);
        \end{tikzpicture}
                \subcaption{
        \label{fig:latent4}Inferred undirected graph with latent node 4.
        }
          \end{subfigure}
    \begin{subfigure}{0.7\columnwidth}
    \centering
            \begin{tikzpicture}
                \tikzstyle{vertex}=[circle,fill=none,minimum size=11pt,inner sep=0pt,thick,draw]
        \tikzstyle{pvertex}=[color=red,circle,fill=white,minimum size=11pt,inner sep=0pt,thick,draw]
          \node[vertex] (n1) {$1$};
          \node[vertex, right of=n1] (n2) {$2$};
          \node[vertex,right of=n2] (n3) {$3$};
          \node[vertex,right of=n3] (n5) {$5$};
          \node[vertex,right of=n5] (n6) {$6$};
          \node[vertex,right of=n6] (n7) {$7$};

          \draw[thick] (n1)--(n2);
          \draw[thick] (n2)--(n3);
          \draw[thick] (n5)--(n6);
          \draw[thick] (n6)--(n7);
          
        \end{tikzpicture}
        \subcaption{
        \label{fig:true edges} Detect true edges.
        }
  \end{subfigure}
  \begin{subfigure}{0.7\columnwidth}
    \centering
  \begin{tikzpicture}
                \tikzstyle{vertex}=[circle,fill=none,minimum size=11pt,inner sep=0pt,thick,draw]
        \tikzstyle{pvertex}=[color=red,circle,fill=white,minimum size=11pt,inner sep=0pt,thick,draw]
          \node[vertex] (n1) {$1$};
          \node[vertex, right of=n1] (n2) {$2$};
          \node[vertex,right of=n2] (n3) {$3$};
          \node[pvertex,right of=n3] (n4) {$4$};
          \node[vertex,right of=n4] (n5) {$5$};
          \node[vertex,right of=n5] (n6) {$6$};
          \node[vertex,right of=n6] (n7) {$7$};

          \draw[thick] (n1)--(n2);
          \draw[thick] (n2)--(n3);
          \draw[blue,thick] (n3)--(n4);
          \draw[blue,thick] (n4)--(n5);
          \draw[thick] (n5)--(n6);
          \draw[thick] (n6)--(n7);
          
        \end{tikzpicture}
                \subcaption{
        \label{fig:trueTopch}Place hidden (corrupt) node 4.
        }
        \end{subfigure}
    \caption{
    \label{fig:hide and infer} This figure shows how Hide and Learn algorithm learns the exact topology of the generative system considered in the example. 
  }
\end{figure}
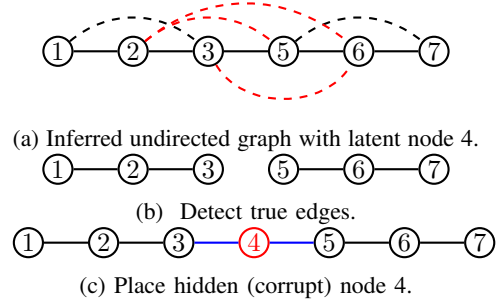

Then, using graphical separation results mentioned in Lemma ~\ref{lemma:separation}, we detect the true edges in the inferred network. This yields two disconnected components as shown in Figure ~\ref{fig:true edges}. Finally, we place the latent node at the point of disconnection and obtain the true generative topology shown in Figure ~\ref{fig:trueTopch}.
\section{Conclusion}\label{sec:conclude}
In this article, we proposed an exact topology learning algorithm for radial bi-directed network of LTI systems in the presence of corruption. We show that for networks where corrupt nodes are three or more hops away from each other deep inside the network, the spurious edges owning to data corruption can be eliminated and the the exact network structure can be determined. We used clique characterization in the inferred undirected graph to determine the set of leaf and corrupt nodes. Then using phase properties of the inverse PSD, we isolated the location of corrupt nodes. Finally,  we hide the corrupt node measurements and adopt \textit{hide and learn} strategy to learn the exact network representation generating the time series observations. We remark here that Algorithm ~\ref{alg:detectCorrupt} and ~\ref{alg:hidelearn} will still work to learn the exact network structure even when there are hidden nodes and corruption simultaneously as long as the location of hidden nodes and the corrupt nodes are at least 3 hops away from each other and at least 3 hops away from the leaf nodes. The future direction of research entails relaxing the assumption on generative topology being a tree. Another pertinent direction would be to quantify the amount of data and provide confidence intervals in estimating PSD from finite samples of data. 

\bibliographystyle{IEEEtran}
\bibliography{ref}
\section*{Appendix: Proof of Theorem \ref{thm:mainalg}}
\emph{Structure Learning with Latent Corrupt Nodes:}
Let $y$ be the time series measurements of corrupt nodes, $Y$, detected after Theorem ~\ref{thm:phase}. 
 Let $o$ denote the set of measurements without $y$. That is, $o=u\setminus Y$. We compute the inverse PSD of $o$. Now, using sparsity pattern in inverse PSD of $o$ as adjacency matrix construct an undirected graph, $\mathcal{T}_m=(V_o,A_o)$. The following result from \cite{talukdar2018topology} characterizes the edges in $\mathcal{T}_m$ and the generative topology $G^T$. 
\begin{lemma}\label{lemma:latentPSD}
Consider a linear dynamical system with generative topology $G^T$. Then, $\Phi _{oo}^{-1}(i,j)(\omega)\ne 0$ for $\omega \in (-\pi,\pi]$, implies that $i$ and $j$ are within four hops of each other in $G^T$.  
\end{lemma}
\emph{True edge set discovery between observed nodes:}
The graph $\mathcal{T}_m$ inferred from Lemma \ref{lemma:latentPSD} contains spurious edges. The objective here is to eliminate the spurious edges and thus identify the true edges. To this, following notion of separation in undirected graphs is introduced.
\begin{definition}[Separation]
Given an undirected graph $G=(V,A)$, the set of nodes $Z\subset V$ is said to \emph{separate} the path between nodes $i$, $j$, if there exist no path between $i$ and $j$ in $G$ after removing the set of nodes $Z$. We denote this by $sep(i,j|Z)$ which is read as $i,j$ are separated by $Z$. 
\end{definition}

\noindent The following result from \cite{talukdar2018topology} provides a topological method based on separation property to identify the observed non-leaf nodes and identify the true edges between them. 
\begin{lemma}\label{lemma:separation}
Suppose $\mathcal{T}_m$ is the graph inferred using measurements $o$ in Lemma ~\ref{lemma:latentPSD}. Suppose there exist observed nodes $c,d$ distinct from observed nodes $a,b$ such that $a-b\in \mathcal{T}_m$. Then, $sep(c,d|\{a,b\})$ holds in $\mathcal{T}_m$ if and only if $a-b$ is a true edge in $G^T$ and $a,b$ are non-leaf nodes.
\end{lemma}
\noindent Combining Lemma \ref{lemma:separation} with the output of Algorithm `\ref{alg:detectCorrupt} that detected the only true edge associated with all the leaf nodes, we have thus identified all true edges associated with the observed nodes. Denote this graph as $\Theta $. 

\emph{Placement of Corrupt Nodes:} The graph $\Theta$ will have multiple radial \textit{disconnected components} denoted as $\theta _j$,  with the disconnections being at the location of the latent corrupt nodes, $Y$. Based on our assumptions, it can be shown that each disconnected component has at least two observed nodes. Thus, for all node $p\in \theta _j$, there is another node $q \in \theta _j$ such that $p-q \in G^T$. Since $G^T$ is a connected graph, the final step is to connect the disconnected components by placing the corrupt nodes at the disconnected locations.
We make use of the prior knowledge gained by inferring the perturbed graph $G^U$ and we map every corrupt node $i\in Y$ to it's corresponding neighborhood  $\nbr _u(i)$ in $G^U$. The following lemma precisely characterizes this.
\begin{lemma}\label{lemma:placecorrupt}
Let $\Theta $ be the disconnected network inferred after removing spurious edges between the observed nodes based on Lemma ~\ref{lemma:separation}. Consider two disconnected components $\theta _1$, $\theta _2$ in $\Theta$ with observed nodes $q\in \theta _1$ and $r\in \theta _2$.  Consider all $p\in \theta _1$ and all $s\in \theta _2$ such that $p-q$ and $r-s$ are edges in $\theta _1$ and  $\theta _2$ respectively. Consider a corrupt node $l\in Y$. Suppose $\{p,q,l,r,s\}$ forms a clique in the perturbed graph, $G^U$. Then, $p-q-l-r-s$ holds in $G^T$ if and only  if
$\angle \Phi ^{-1}_{uu}(\omega)(p,s)$ is a constant for all $\omega \in (-\pi, \pi]$.
\end{lemma}
\begin{proof}
Since $\{p,q,l,r,s\}$ forms a clique in $G^U$ and $G^T$ is a tree, it follows that $l$ is located at the point of disconnection between $p-q$ and $r-s$. What needs to be shown is the correct alignment among the paths $q-p-l-r-s$, $p-q-l-s-r$ and $q-p-l-s-r$ and $p-q-l-r-s$ in $G^T$. To this we will analyze the phase of inverse PSD entry corresponding to pairs from $\{p,q\}\times \{r,s\}$ described in Proposition ~\ref{prop:relaxeddistphase}.    
Before that we will need the following proposition.
\begin{proposition}\label{prop:localPsi}
Suppose $p-q-l-r-s$ holds in $G^T$ where $l$ is a corrupt node. Then, for any $a\in \{p,q\}$ and $b\in \{r,s\}$, $\Phi ^{-1}_{uu}(\omega)(a,b)=\Psi ^{-1}_{1}(\omega)(a,b)$, where $\Psi ^{-1}_1$ is defined by \eqref{eq:invPsiInduction} and $v_1=l$.
\end{proposition}
\begin{proof}
For $k=2,3,\dots , n$ perturbed nodes we will inductively show that $\Psi ^{-1}_k(a,b)=\Psi ^{-1}_1(a,b)$. 

We will require the following claim: 
for any $v_k$ being a perturbed node ($k>1$), at the most only one of $a-v_k$ or $v_k-b$ holds in $G^T$. Note that there is already a path, $p-q-l-r-s$, consisting of $a$ and $b$.  As $v_k\ne l$ and $G^T$ is a tree, existence of $a-v_k-b$ violates the assumption that $G^T$ is a tree. This proves the claim. In other words, at least one of $a-v_k$ or $v_k-b$ does not hold in $G^T$. Suppose $v_k-b$ does not hold true. (The case $a-v_k$ can be shown similarly). Refer to this result as $R1$.  

Consider $k=2$. We will show that $\Psi^{-1}_2(a,b)=\Psi^{-1}_1(a,b)$. Using ~\eqref{eq:invPsiInduction},   
$\Psi ^{-1} _{2}(a,b)=\Psi ^{-1}_{1}(a,b)-\Gamma _{2}(a,b)$. Note that $\Gamma _2(a,b)=\Psi ^{-1}_1(a,v_2)\Psi ^{-1}_1(v_2,b)\Delta ^{-1}_{2}$. We will show that $\Gamma _2(a,b)=0$. Using ~\eqref{eq:invPsiInduction}, $\Psi ^{-1} _{1}(v_2,b)=\Psi ^{-1}_{0}(v_2,b)-\Gamma _{1}(v_2,b)$ where $\Gamma_1(v_2,b):= \Psi^{-1} _{0}(v_2,v_1)\Psi
^{-1}_{0}(v_1,b)\Delta ^{-1}_{1}.$ By ~\eqref{eq:psi0}, $\Psi ^{-1}_0(v_1,v_2)=\frac{\Phi ^{-1}_{xx}(v_1,v_2)}{\overline{h_{v_1}}h_{v_2}}$. 
As $v_1$ and $v_2$ are at least 3-hops away in $G^T$, using ~\eqref{eq:phiXinv}, $\Phi ^{-1}_{xx}(v_1,v_2)=0$. 
 Thus, $\Gamma _1(v_2,v_1)=0.$  
  Invoking $R1$, we have that $v_2-b$ does not hold in $G^T$. Then, $v_2$ can either be a 2-hop neighbor of $b$ in $G^T$ or not. 
  
  Suppose $v_2$ is a 2-hop neighbor of $b$. 
  Then, $v_2$ cannot be a 1-hop or 2-hop neighbor of $a$ because this leads to two paths connecting $a$ and $b$: one through $v_2$ and the other being $p-q-l-r-s$, violating the condition that $G^T$ is a tree. Thus, $v_2$ is neither a 1-hop neighbor nor a 2-hop neighbor of $a$ in $G^T$. 
  Thus, using ~\eqref{eq:phiXinv} we have that $\Phi ^{-1}_{xx}(a,v_2)=0$. By ~\eqref{eq:psi0}, $\Psi ^{-1}_0(a,v_2)=\frac{\Phi ^{-1}_{xx}(a,v_2)}{\overline{h_{a}}h_{v_2}}$. This implies $\Psi _{0}^{-1}(a,v_2)=0$. Therefore, $\Psi _{1}^{-1}(a,v_2)=0$ and hence $\Gamma _2(a,b)=0$. Thus we have proved that $\Psi ^{-1} _{2}(a,b)=\Psi ^{-1}_{1}(a,b)$.  
  
Now consider that $v_2$ is not a 2-hop neighbor of $b$.  Then, using ~\eqref{eq:phiXinv} we have that $\Phi ^{-1}_{xx}(v_2,b)=0$. By ~\eqref{eq:psi0}, $\Psi ^{-1}_0(v_2,b)=\frac{\Phi ^{-1}_{xx}(v_2)}{\overline{h_{v_2}}h_{b}}$. This implies $\Psi _{0}^{-1}(v_2,b)=0$. Therefore, $\Psi _{1}^{-1}(v_2,b)=0$ and hence $\Gamma _2(a,b)=0$. Thus we have proved that $\Psi ^{-1} _{2}(a,b)=\Psi ^{-1}_{1}(a,b)$.

Now assume that the claim holds for some $k>2$. That is, $\Psi ^{-1}_k(a,b)=\Psi ^{-1}_1(a,b).$ Using ~\eqref{eq:invPsiInduction},   
$\Psi ^{-1} _{k+1}(a,b)=\Psi ^{-1}_{k}(a,b)-\Gamma _{k+1}(a,b)$. 
As shown in Theorem ~\ref{thm:phase}, $\Psi ^{-1}_{k}(a,v_{k+1})=\Psi ^{-1}_{1}(a,v_{k+1})$ and $\Psi ^{-1} _{k}(v_{k+1},b)=\Psi ^{-1}_1 (v_{k+1},b)$. Invoking $R1$, we have that $v_{k+1}-b$ does not hold in $G^T$. Then, $v_{k+1}$ can either be a 2-hop neighbor of $b$ in $G^T$ or not. 
  
  Suppose $v_{k+1}$ is a 2-hop neighbor of $b$. Similar to argument for $v_2$, $v_{k+1}$ cannot be a 1-hop or 2-hop neighbor of $a$. As $v_1$ and $v_{k+1}$ are at least 3-hops away in $G^T$, using ~\eqref{eq:phiXinv}, $\Phi ^{-1}_{xx}(v_1,v_{k+1})=0$. Now consider that $v_{k+1}$ is not a 2-hop neighbor of $b$. Similar to argument for $v_2$, $\Psi _{1}^{-1}(v_{k+1},b)=0$ and hence $\Psi _{k}^{-1}(v_{k+1},b)=0$ and $\Gamma _{k+1}(a,b)=0$. Thus we have proved that $\Psi ^{-1} _{k+1}(a,b)=\Psi ^{-1}_{1}(a,b)$. 
\end{proof}
We now proceed to stating and proving Proposition ~\ref{prop:relaxeddistphase}.
\begin{proposition} \label{prop:relaxeddistphase}
\it{If $p-q-l-r-s$ holds in $G^T$, then $\angle \Phi ^{-1}_{uu}(\omega)(p,s)$ is a constant while $\angle \Phi ^{-1}_{uu}(\omega)(p,r)$, $\angle \Phi ^{-1}_{uu}(\omega)(q,s)$ and $\angle \Phi ^{-1}_{uu}(\omega)(q,r)$ are non-constant for all $\omega \in [-\pi, \pi]$.}
\end{proposition}
\begin{proof}  
Using Proposition ~\ref{prop:localPsi} we have that for any $a\in \{p,q\}$ and $b\in \{r,s\}$, $\Phi ^{-1}_{uu}(\omega)(a,b)=\Psi ^{-1}_{1}(\omega)(a,b)$, where $\Psi ^{-1}_1$ is defined by \eqref{eq:invPsiInduction} and $v_1=l$.

As $v_1=l$ is the only corrupt node, $\Phi ^{-1}_{uu}(\omega)(a,b)$ can be expressed as:
\begin{equation}\label{eq:PSentry}
    \Phi ^{-1} _{uu}(a,b)=\Psi ^{-1}_{0}(a,b)-\Psi^{-1} _{0}(a,l)\Psi
^{-1}_{0}(l,b)\Delta ^{-1}_{v_1}.
\end{equation} 
Moreover, as $a,b$ are not corrupt nodes, we have that $h_a(\omega)=h_b(\omega)=1$. Thus, $\Psi^{-1} _{0}(a,l) =\Phi^{-1} _{xx}(a,l)$ and $\Psi^{-1} _{0}(l,b) =\Phi^{-1} _{xx}(l,b)$. Now we will show that the term $\Delta ^{-1}_{v_1}$ is real valued for all $\omega$.
Now, $\Delta_{l}=d_{v_1}^{-1}+\Psi ^{-1}_{0}(v_1,v_1)$. By ~\eqref{eq:psi0} we have that $\Psi ^{-1}_{0}(v_1,v_1)=\frac{\Phi ^{-1}_{xx}(v_1,v_1)}{|h_l(\omega)|^2}$. Using ~\eqref{eq:phiXinv}, we have that $\Phi ^{-1}_{xx}(v_1,v_1)$ is real valued and therefore, $\Psi ^{-1}_{0}(v_1,v_1)$ is real valued. As $d(\omega)$ is the PSD of autocorrelation of a WSS process it will be real and non-negative valued. A more formal description is in \cite{SLSarXivlinear}. Thus, $\Delta ^{-1}_{v_1}$ is real valued for all $\omega \in [- \pi, \pi]$.

We now proceed to evaluating $\angle \Phi ^{-1}_{uu}(\omega)(a,b)$ for all combinations of $a\in \{p,q\}$ and $b\in\{r,s\}$.
  
  \emph{$(q,r)$}: In this case $q,r$ are 1-hop neighbors of $l$ in $G^T$ and hence as discussed in Theorem ~\ref{thm:phase}, $\Psi ^{-1}_0(q,l)$ and $\Psi ^{-1}_0(l,r)$ will be non-constant transfer functions. Thus, $\Phi _{uu}^{-1}(q,r)(\omega)$
will be non-constant transfer functions. 

  \emph{$(p,r)$}: Here, $s$ is a 1-hop neighbor of $l$ in $G^T$.  Thus as discussed in Theorem ~\ref{thm:phase},  $\Psi ^{-1}_0(l,s)$ will be non-constant transfer function. Thus, $\Phi _{uu}^{-1}(p,r)(\omega)$
will be a non-constant transfer function. The case $(q,s)$ can be shown similarly where $r$ is a 1-hop neighbor of $l$ in $G^T$..

\emph{$(p,s)$}: As $p,s$ are 2-hop neighbors of $l$ in $G^T$. Then, by ~\eqref{eq:phiXinv} we have that $\Phi _{xx}^{-1}(p,l)$ and $\Phi _{xx}^{-1}(l,s)$ being real valued. Thus, $\Psi ^{-1}_0(p,l)$ and $\Psi ^{-1}_0(l,s)$ will be real valued transfer functions. As $p,s$ are 4 hop neighbors, using ~\eqref{eq:phiXinv}  $\Phi _{xx}^{-1}(p,s)=0$. This implies $\Psi ^{-1}_0(p,s)=0$.  Then, $ \Phi ^{-1}_{uu}(\omega)(p,s)=-\Psi ^{-1}_0(p,l)\Psi ^{-1}_0(l,s)\Delta ^{-1}_l$. We have $\Psi ^{-1}_0(p,l)$,  $\Psi ^{-1}_0(l,s)$ and $\Delta ^{-1}_l$ being real valued. Thus, $\Phi _{uu}^{-1}(p,s)(\omega)$
will be a real valued transfer functions. Therefore, $\angle \Phi _{uu}^{-1}(p,s)(\omega)$ will be a constant for all $\omega \in [-\pi, \pi].$ 
\end{proof}
It follows from above lemma that only if the corresponding phase properties hold in $p-q-l-r-s$, then it is the only correct alignment as any other alignment will have non-2 hop neighbors as 4 hops away and hence will violate the constant phase argument. This verifies Lemma ~\ref{lemma:placecorrupt}.
\end{proof}
\end{document}